\documentclass{llncs}
\usepackage{amsmath}
\usepackage{amssymb}
\usepackage{graphicx,xspace,tikz}
\usetikzlibrary{arrows,automata}

\newcommand{\leaveout}[1]{}

\newcommand{\sosrule}[2]{\frac{\raisebox{.7ex}{\normalsize{$#1$}}}
                        {\raisebox{-1.0ex}{\normalsize{$#2$}}}}

\newcommand{\ltsntrans}[1]{\,{\stackrel{{#1}}{\nrightarrow}}\,}
\newcommand{\ltstrans}[1]{\xrightarrow{#1}}
\newcommand{\kstrans}{\to}

\newcommand{\bisim}{\mbox{$\underline{\leftrightarrow}$}}

\newcommand{\dsbbisim}{\mbox{$\underline{\leftrightarrow}_{\mathrm{dsb}}$}}

\newcommand{\stuttering}{\,\approx_{\mathrm{s}}\,}

\newcommand{\lts}{\mathsf{lts}}
\newcommand{\ltsrev}{\lts^{-1}}
\newcommand{\ks}{\mathsf{ks}}
\newcommand{\ksrev}{\ks^{-1}}

\newcommand{\ie}{\emph{i.e.}}
\newcommand{\eg}{\emph{e.g.}}
\newcommand{\viz}{\emph{viz.}}
\newcommand{\etal}{\emph{et al}\xspace}

\title{Folk Theorems on the Correspondence between State-Based and Event-Based Systems}
\author{Michel A. Reniers\inst{1} \and Tim A.C. Willemse\inst{2}
}
\institute{
Department of Mechanical Engineering, Eindhoven University of Technology,
\\ {P.O.~Box~513}, NL-5600~MB~~Eindhoven, The Netherlands
\and
Department of Computer Science, Eindhoven University of Technology,
\\ {P.O.~Box~513}, NL-5600~MB~~Eindhoven, The Netherlands
}

\begin{document}
\begin{frontmatter}

\maketitle

\begin{abstract}
Kripke Structures and Labelled Transition Systems are the two most
prominent semantic models used in concurrency theory. Both models
are commonly believed to be equi-expressive. One can find many ad-hoc
embeddings of one of these models into the other. We build upon the
seminal work of De Nicola and Vaandrager that firmly established the
correspondence between stuttering equivalence in Kripke Structures
and divergence-sensitive branching bisimulation in Labelled Transition
Systems. We show that their embeddings can also be used for a range of
other equivalences of interest, such as strong bisimilarity, simulation
equivalence, and trace equivalence. Furthermore, we extend the results by
De Nicola and Vaandrager by showing that there are additional translations
that allow one to use minimisation techniques in one semantic domain to
obtain minimal representatives in the other semantic domain for these
equivalences.

\end{abstract}

\end{frontmatter}

\pagenumbering{arabic}

\newcommand{\KS}{\textsf{KS}\xspace}
\newcommand{\LTS}{\textsf{LTS}\xspace}

\section{Introduction}\label{Sect:intro}

Concurrency theory, and process theory in general, deal with the analysis
and specification of behaviours of reactive systems, \ie, systems that
continuously interact with their environment. Over the course of the
past decades, a rich variety of formal languages have been proposed
for modelling such systems effectively. At the level of the semantics,
however, consensus seems to have been reached over the models used to
represent these behaviours. Two of the most pervasive models are the
state-based model generally referred to as \emph{Kripke Structures}
and the event-based model known as \emph{Labelled Transition Systems},
henceforth referred to as \KS and \LTS.

The common consensus is that both the \KS and \LTS models are on
equal footing. This is supported by several embeddings of one model
into the other that have been studied in the past, see below for a
brief overview of the relevant literature.  As far as we have been
able to trace, in all cases embeddings of both semantic models were
considered modulo a single behavioural equivalence. For instance, in
their seminal work~\cite{DBLP:journals/jacm/NicolaV95}, De Nicola and
Vaandrager showed that there are embeddings in both directions showing
that stuttering equivalence~\cite{DBLP:journals/tcs/BrowneCG88}
in \KS coincides with divergence-sensitive branching
bisimulation~\cite{vanGlabbeek96} in \LTS. The embeddings, however,
look a bit awkward from the viewpoint of concrete equivalence relations.

On the basis of these results, one cannot arrive at the conclusion that
the embeddings also work for a larger set of equivalences. For instance,
it is very easy to come up with a mapping that reflects and preserves
branching-time equivalences while breaking linear-time equivalences, by
exposing observations of branching through the encodings.  Note that it is
equally easy to construct encodings that break branching-time equivalences
while reflecting and preserving some linear-time equivalences, \eg,
by including some form of determinisation in the embeddings.

Our contributions are as follows. Using the
\KS-\LTS embeddings $\lts{}$ and $\ks{}$ of De Nicola and
Vaandrager in~\cite{DBLP:conf/litp/NicolaV90}, in Section \ref{Sect:preservations} we formally establish the
following relations under these embeddings:
\begin{enumerate}
\item bisimilarity in \KS reflects and preserves bisimilarity in \LTS;
\item similarity in \KS reflects and preserves similarity in \LTS;
\item trace equivalence in \KS reflects and preserves completed trace equivalence
in \LTS.
\end{enumerate}
These results add to the credibility that indeed both worlds are on
equal footing, and it may well be that the embeddings $\ks$ and $\lts$
are in fact canonical.

As already noted in~\cite{DBLP:conf/litp/NicolaV90}, there is no
immediate correspondence between the embeddings $\lts{}$ and $\ks{}$.
For instance, one cannot move between \KS and \LTS and back again by
composing $\lts{}$ and $\ks{}$. We mend this situation by introducing two
additional translations, \viz, $\ltsrev$ and $\ksrev$, that can be used
to this end. Moreover, we show that combining these with the original
embeddings enables one to minimise with respect to an equivalence in
\KS by minimising the embedded artefact in \LTS (and \emph{vice versa}).

From a practical point of view, our contributions allow one to smoothly
move between both semantic models using a single set of translations.
This reduces the need for implementing dedicated software in one setting
when one can take advantage of state-of-the-art machinery available in
the other setting.

\paragraph{Related Work}

In their seminal paper (see~\cite{DBLP:journals/jacm/NicolaV95})
on logics for branching bisimilarity, De Nicola and Vaandrager
established, among others, a firm correspondence between the
divergence-sensitive branching bisimilarity of~\cite{vanGlabbeek96},
and stuttering equivalence~\cite{DBLP:journals/tcs/BrowneCG88}.
Their results spawned an interest in the relation
between temporal logics in the \LTS and the \KS setting, see
\eg~\cite{DBLP:journals/cn/NicolaFGR93,DBLP:conf/litp/NicolaV90}. The
latter both contain the embeddings that we use in this paper, differing
slightly from the ones proposed in~\cite{DBLP:journals/jacm/NicolaV95},
which in turn were in part inspired by the (unpublished) embedding
by Emerson and Lei~\cite{EL:84}.  The tight correspondence between
stuttering equivalence and branching bisimilarity that was exposed,
led Groote and Vaandrager to define algorithms for deciding said
equivalences in~\cite{DBLP:conf/icalp/GrooteV90}. Their algorithms
(and their correctness proofs), however, are stated directly in terms of
the appropriate setting, and do not appear to use the embeddings $\lts$
and $\ks$ (but they might have acted as a source of inspiration).

Apart from the few documented cases listed above, many ad-hoc
embeddings are known to work for equivalences that are not sensitive to
abstraction. For instance, one can model the state labelling in a Kripke
Structure by means of labelled self-loops, or directly on the edges to
the next states, thereby exposing the same information. Such embeddings,
however, fail for equivalences that are sensitive to abstraction, such
as stuttering equivalence, which basically compresses sequences of states
labelled with the same state information.

\paragraph{Outline} In Section~\ref{Sect:preliminaries}, we formally
introduce the computational models \KS and \LTS, along with the
embeddings $\ks$ and $\lts$. The latter are proved to preserve and
reflect the additional three pairs of equivalences relations stated
above. In Section~\ref{sec:minimisations}, we introduce the inverses
$\ksrev$ and $\ltsrev$, and we show that these can be combined with
$\ks$ and $\lts$, respectively, to obtain our minimisation results.
We finish with a brief summary of our contributions and an outlook to
some interesting open issues.

\section{Preliminaries}\label{Sect:preliminaries}

Central in both models of computation that we consider, \ie, \KS and
\LTS, are the notions of \emph{states} and \emph{transitions}. While the
\KS model emphasises the information contained \emph{in} such states,
the \LTS model emphasises the \emph{state changes} through some action
modelling a real-life event. Let us first recall both models of computation.

\newcommand{\AP}{\ensuremath{AP\xspace}}
\newcommand{\tuple}[1]{\ensuremath{\langle\,{#1}\,\rangle}}
\begin{definition}
A \emph{Kripke Structure} is a structure $\tuple{S, \AP, \to, L}$, where
\begin{itemize}
\item $S$ is a set of states;

\item $\AP$ is a set of atomic propositions;

\item $\to \subseteq S \times S$ is a total transition relation, \ie, for
all $s \in S$, there exists $t \in S$, such that $(s,t) \in\to$;

\item $L : S \to 2^{\AP}$ is a state labelling.
\end{itemize}
\end{definition}
By convention, we write $s \to t$ whenever $(s,t) \in \to$.

\begin{remark}
The transition relation in the \KS model is traditionally required to be
total. Our results do not depend on the requirement of totality, but
we choose to enforce totality in favour of a smoother presentation and
more concise definitions. Without totality,
slightly more complicated treatments of the notions of paths and traces
(see also Section~\ref{sec:traces}) are needed.
\end{remark}
With the above restriction in mind, we define the \LTS model with
a similar restriction imposed on it.

\newcommand{\act}{\ensuremath{A}ct}
\begin{definition}[Labelled Transition System]
A structure $\tuple{S,\act,\ltstrans{}}$ is an \LTS, where:
\begin{itemize}
\item $S$ is a set of states;
\item $\act$ is a set of actions;
\item $\ltstrans{} \subseteq S \times (\act \cup \{\tau\}) \times S$ is a total
transition relation,
\ie, for all $s \in S$, there are $a \in \act$, $t \in S$, such that
$(s,a,t) \in \ltstrans{}$.
\end{itemize}
\end{definition}
In lieu of the convention for \KS, we write $s \ltstrans{a} t$ whenever
$(s,a,t) \in \ltstrans{}$.

Note that in the setting of the \LTS model, a special constant $\tau$
is assumed outside the alphabet of the set of actions $\act$ of any
concrete transition system. This constant is used to represent so-called
silent steps in the transition system, modelling events that are
unobservable to any witness of the system.

In~\cite{DBLP:conf/litp/NicolaV90}, De Nicola and Vaandrager considered
embeddings called $\lts$ and $\ks$, which allowed one to move from \KS
models to \LTS models, and, \emph{vice versa}, from \LTS models to \KS
models. We repeat these embeddings below, starting with the embedding
from \KS into \LTS.

\newcommand{\shadow}[1]{\ensuremath{\bar{#1}}}
\begin{definition}
\label{translationlts}
The embedding $\lts : \KS \to \LTS$ is defined as $\lts(K) =
\tuple{S', \act, \ltstrans{}}$ for arbitrary Kripke Structures $K =
\tuple{S,\AP,\kstrans{},L}$, where:

\begin{itemize}
\item $S' = S \cup \{ \shadow{s} ~|~ s \in S \}$, where
it is assumed that $\shadow{s} \notin S$ for all $s \in S$;
\item $\act = 2^{\AP} \cup \{\bot\}$;
\item $\ltstrans{}$ is the smallest relation satisfying:
$$
\begin{array}{cp{1cm}c}
\sosrule{}{s \ltstrans{\bot} \shadow{s}} & &
\sosrule{s \kstrans{} t \qquad L(s) = L(t)}{s \ltstrans{\tau} t} \\
\\
\sosrule{}{\shadow{s} \ltstrans{L(s)} s} & &
\sosrule{s \kstrans{} t \qquad L(s) \not= L(t)}{s \ltstrans{L(t)} t} \\
\end{array}
$$
\end{itemize}
\end{definition}
The fresh symbol $\bot$ is used to signal a forthcoming encoding of the state information of the Kripke Structure. Encoding the state information by means of a self-loop $s \ltstrans{L(s)} s$ introduces problems in preserving and reflecting equivalences that are sensitive to abstraction.

\begin{definition}
\label{translationks}
The embedding $\ks : \LTS \to \KS$ is formally defined as
$\ks(T) = \tuple{S', \AP, \kstrans{}, L}$ for Labelled Transition
System $T = \tuple{S, \act, \ltstrans{}}$, where:
\begin{itemize}
\item $S' = S \cup \{ (s,a,t) \in \ltstrans{} ~|~ a \not=\tau \}$;
\item $\AP = \act \cup \{\bot\}$, where $\bot \notin \act$;
\item $\kstrans$ is the least relation satisfying:
$$
\begin{array}{cp{1cm}cp{1cm}c}
\sosrule{}{s \kstrans (s,a,t)}
&
&
\sosrule{}{(s,a,t) \kstrans t}
&
&
\sosrule{s \ltstrans{\tau} t}{s \kstrans t}
\end{array}
$$
\item $L(s) = \{\bot\}$ for $s \in S$, and $L((s,a,t)) = \{a\}$.
\end{itemize}
\end{definition}
In this embedding the fresh symbol $\bot$ is used to label the states from the Labelled Transition System.
The reason to treat $\tau$-transitions different from ordinary actions is that otherwise equivalences that abstract from sequences of $\tau$-transitions are not reflected well.

Observe that, as already stated in~\cite{DBLP:conf/litp/NicolaV90},
due to the artefacts introduced by the embeddings, moving from \LTS
to \KS and back again yields transition systems incomparable to the
original ones. Consequently, in \LTS, one cannot take advantage of tools
for minimising in the setting of \KS, and \emph{vice versa}. We defer
further discussions on this matter to Section~\ref{sec:minimisations}.

\section{Preservations and Reflections of Equivalences Under $\lts$ and $\ks$}
\label{Sect:preservations}

The embeddings $\lts$ and $\ks$ have already been shown to preserve and
reflect stuttering equivalence~\cite{DBLP:journals/tcs/BrowneCG88} and divergence-sensitive
branching bisimulation~\cite{vanGlabbeek96} by De Nicola and Vaandrager. In this
section, we introduce three additional pairs of equivalences and show
that these are also preserved by the embeddings $\lts$ and $\ks$. Our
choice for these four equivalences is motivated largely by the limited
set of equivalence's available in the \KS model (contrary to the \LTS model,
which offers a very fine-grained lattice of equivalence relations).

\begin{remark}
For reasons of brevity, throughout this paper we define equivalence
relations on states within a single \LTS (resp.\ \KS) rather than
equivalence relations between different models in \LTS (resp.\ \KS). Note
that this does not incur a loss in generality, as it is easy to define
the latter in terms of the former.

\end{remark}

\subsection{Similarity}
\label{sec:similarity}

Both \KS and \LTS have well-developed theories revolving around
similarity. We first formally define both notions.

\newcommand{\simm}{\simeq}
\begin{definition}
Let $K = \tuple{S,\AP,\kstrans,L}$ be a Kripke Structure. A relation
$B \subseteq S \times S$ is a \emph{simulation relation} iff for every
$s,s' \in S$ satisfying $(s,s') \in B$:
\begin{itemize}
\item $L(s) = L(s')$;
\item for all $t \in S$, if $s \kstrans t$, then $s' \kstrans t'$ for
some $t' \in S$ such that $(t,t') \in B$.

\end{itemize}
For states $s,s' \in S$, we say $s$ is simulated by $s'$ if there is
a simulation relation $B$, such that $(s,s') \in B$. States $s,s' \in
S$ are said to be \emph{similar}, denoted $K \models s \simm s'$ iff
there are simulation relations $B$ and $B'$, such that $(s,s') \in B$
and $(s',s) \in B'$.

\end{definition}
\begin{remark}
It should be noted that when lifting our notion of similarity to
an equivalence relation between different models in \KS,
the first requirement is sometimes stated as $L(s) = L'(s') \cap \AP$,
where $L'$ is the state labelling of the second \KS model, and $\AP$
is the set of atomic propositions of the first \KS model. In this
case, some form of abstraction is included already, and care should be
taken to deal with such abstractions properly when lifting all our results to
such a setting.
\end{remark}

\begin{definition} Let $T = \tuple{S, \act, \ltstrans{}}$ be a Labelled
Transition System. A relation $B \subseteq S \times S$ is a
\emph{simulation relation} iff for every $s,s' \in S$ satisfying
$(s,s') \in B$:
\begin{itemize}
\item for all $t \in S$ and $a \in \act \cup \{\tau\}$,
if $s \ltstrans{a} t$, then $s' \ltstrans{a} t'$
for some $t' \in S'$ such that $(t,t') \in B$.

\end{itemize}
State $s \in S$ is said to be simulated by state $s' \in S$ if there is
a simulation relation $B$, such that $(s,s') \in B$. States $s,s' \in S$
are \emph{similar}, denoted
$T \models s \simm s'$ iff there are simulation relations $B$ and $B'$, such
that $(s,s') \in B$ and $(s',s) \in B'$.

\end{definition}

The theorems below state that indeed, embedding $\lts$ preserves
and reflects \KS-similarity through \LTS-similarity (see
Theorem~\ref{th:lts_similarity}), and \emph{vice versa}, embedding
$\ks$ preserves and reflects \LTS-similarity through \KS-similarity
(Theorem~\ref{th:ks_similarity}).

\begin{theorem}
\label{th:lts_similarity}
Let $K = \tuple{S, \AP,\kstrans, L}$ be an arbitrary Kripke Structure.
Then, for all $s, s' \in S$, we have
$K \models s \simm s'$ iff $\lts(K) \models s \simm s'$.
\end{theorem}

\begin{proof}
See Appendix \ref{pf:th:lts_similarity}.\qed
\end{proof}

\begin{theorem}
\label{th:ks_similarity}
Let $T = \tuple{S, \act, \ltstrans{}}$ be a Labelled Transition System.
Then for all $s,s' \in S$, we have $T \models s \simm s'$ iff
$\ks(T) \models s \simm s'$.
\end{theorem}

\begin{proof}
See Appendix \ref{pf:th:ks_similarity}.\qed
\end{proof}

\subsection{Bisimilarity}
\label{sec:bisimilarity}

A slightly stronger notion of equivalence that is rooted in the same concepts
as similarity, is \emph{bisimilarity}. Again, bisimilarity has been defined
in both \KS and \LTS, and we here show that both definitions agree through
the embeddings $\lts$ and $\ks$.

\begin{definition}
Let $K = \tuple{S, \AP, \kstrans, L}$ be a Kripke Structure. States
$s,s' \in S$ are said to be \emph{bisimilar}, denoted
$K \models s \bisim s'$ iff there is a \emph{symmetric
simulation relation} $B$, such that $(s,s') \in B$.

\end{definition}
Similarly, we define bisimilarity in the setting of \LTS as follows:

\begin{definition}
Let $T = \tuple{S, \act, \ltstrans{}}$ be a Labelled Transition System.
States $s,s' \in S$ are \emph{bisimilar}, written $T \models s \bisim s'$ iff
there is a \emph{symmetric simulation relation} $B$, such that
$(s,s') \in B$.

\end{definition}

\begin{theorem}
\label{th:lts_bisimilarity}
Let $K = \tuple{S, \AP, \kstrans, L}$ be a Kripke Structure. Then for all
$s,s' \in S$, we have $K \models s \bisim s'$ iff $\lts(K) \models s \bisim s'$.
\end{theorem}

\begin{proof}
The proof is along the lines of the proof for similarity. For details,
see Appendix \ref{pf:th:lts_bisimilarity}.\qed
\end{proof}

\begin{theorem}
Let $T = \tuple{S, \act, \ltstrans{}}$ be a Labelled Transition System.
For all $s, s' \in S$, we have $T \models s \bisim s'$ iff $
\ks(T) \models s \bisim s'$.
\end{theorem}

\begin{proof} Again, the proof is along the lines of the proof for
similarity. \qed
\end{proof}

\subsection{Stuttering Equivalence -- Divergence-Sensitive Branching
Bisimilarity}
\label{sec:stuttering}

In this section, we merely repeat the definitions for stuttering
equivalence and divergence-sensitive branching bisimilarity. In
Section~\ref{sec:minimisations}, we come back to these equivalence
relations and state several new results for these.\\

The following definition for stuttering equivalence is taken
from \cite{DBLP:journals/jacm/NicolaV95}, where it is shown to coincide
with the original definition by Brown, Clarke and Grumberg~\cite{DBLP:journals/tcs/BrowneCG88}.
We prefer the former phrasing because of its coinductive nature.
\begin{definition}
Let $K = \tuple{S, \AP, \kstrans, L}$ be a Kripke Structure. A
\emph{symmetric} relation $B \subseteq S \times S$ is a
\emph{divergence-blind stuttering equivalence} iff for all $(s,s') \in B$:
\begin{itemize}
\item $L(s) = L(s')$;
\item for all $t \in S$, if $s \kstrans t$, then there exist
$s_0',\ldots, s_n' \in S$, such that $s' = s_0'$ and $(t,s'_n) \in B$, and
for all $i < n$, $s_i' \kstrans s_{i+1}'$ and $(s,s_i') \in B$.
\end{itemize}
\end{definition}
\begin{definition} Let $K = \tuple{S, \AP, \kstrans, L}$ be a Kripke
Structure. Let the Kripke Structure $K_d = \tuple{S_d, \AP_d,
\kstrans_d, L_d}$ be defined as follows:
\begin{itemize}
\item $S_d = S \cup \{s_d\}$ for some fresh state $s_d \notin S$;

\item $\AP_d = \AP \cup \{d\}$ for some fresh proposition $d \notin \AP$;

\item $\kstrans_d = \kstrans \cup \{(s,s_d) ~|~ \text{$s$ is on an
infinite path of states labelled $L(s)$, or $s=s_d$}\}$;

\item for all $s \in S$, $L_d(s) = L(s)$, and $L_d(s_d) = \{d\}$.

\end{itemize}
States $s, s' \in S$ are said to be \emph{stuttering equivalent},
notation: $K \models s \stuttering s'$ iff there is a divergence-blind
stuttering equivalence relation $B$ on $S_d$ of $K_d$, such that
$(s,s') \in B$.
\end{definition}
The origins of divergence-sensitive branching bisimilarity can be
traced back to~\cite{vanGlabbeek96}.
In~\cite{DBLP:journals/fuin/GlabbeekLT09}, Van Glabbeek
\etal demonstrate that various incomparable phrasings of
the divergence property all coincide with the original definition.
For our purposes the following formulation is most suitable.
\newcommand{\nat}{\ensuremath{\mathbb{N}}}
\begin{definition}
Let $T = \tuple{S, \act, \ltstrans{}}$ be a Labelled Transition System.
A symmetric relation $B \subseteq S \times S'$ is a \emph{divergence-sensitive
branching simulation relation} iff for all $(s,s') \in B$:
\begin{itemize}

\item if there is an infinite sequence of states $s_0\ s_1\ s_2 \cdots$
such that $s = s_0$ and $s_i \ltstrans{\tau} s_{i+1}$ for all $i$, then
there exist a mapping $\sigma : \nat \to \nat$, and an infinite sequence of
states $s'_0\ s'_1\ s'_2 \cdots$ such that $s' = s'_0$, $s'_k \ltstrans{\tau}
s'_{k+1}$ and $(s_{\sigma(k)}, s_k') \in B$ for all $k \in \nat$;

\item for all $t \in S$ and $a \in \act \cup \{\tau\}$, if
$s \ltstrans{a} t$, then $a=\tau$ and $(t,s') \in B$, or
$s' \ltstrans{\tau^*} s^* \ltstrans{a} t'$ for some $s^*, t' \in S$
such that $(s,s^*) \in B$ and  $(t,t') \in B$.

\end{itemize}
States $s, s' \in S$ are divergence-sensitive branching bisimilar,
notation $s \dsbbisim s'$ iff there is a symmetric divergence-sensitive
branching simulation relation $B$, such that $(s,s') \in B$.
\end{definition}

\subsection{Trace Equivalence -- Completed Trace Equivalence}
\label{sec:traces}

Trace equivalence and completed trace equivalence are the only
linear-time equivalence relations that we consider in this paper. In
defining these equivalence relations, we require some auxiliary notions,
basically defining what a \emph{computation} is in our respective
models of computation.

\newcommand{\paths}[1]{\ensuremath{\mathsf{Paths}(#1)}}
\begin{definition} Let $K = \tuple{S, \AP, \kstrans, L}$ be a Kripke
Structure. A \emph{path} starting in state $s \in S$ is an infinite
sequence $s_0\ s_1\ \ldots$, such that $s_i \kstrans{} s_{i+1}$ for
all $i$, and $s = s_0$. The set of all paths starting in $s$ is denoted
$\paths{s}$.
\end{definition}
Basically, a path formalises how a single computation evolves in
time. Actually, it is the information contained in the states that are
visited along such a computation that is often of interest, as it shows
how the state information evolves in time. This is exactly captured
by the notion of a \emph{trace}.

\newcommand{\trace}[1]{\ensuremath{\mathsf{Trace}(#1)}}
\newcommand{\traces}[1]{\ensuremath{\mathsf{Traces}(#1)}}
\newcommand{\traceeq}{\ensuremath{\simeq_{\mathrm{t}}}}

\begin{definition} Let $K = \tuple{S, \AP, \kstrans, L}$ be a Kripke
Structure. Let $\pi = s_0\ s_1\ \ldots$ be a path starting in $s_0$. The
\emph{trace} of $\pi$, denoted $\trace{\pi}$, is the infinite sequence
$L(s_0)\ L(s_1)\ \ldots$.  For a set of paths $\Pi$, we set
$$\traces{\Pi} = \{ \trace{\pi} \mid \pi \in  \Pi\}$$
States $s,s' \in S$ are \emph{trace equivalent}, denoted $K \models
s \traceeq s'$, if $\traces{\paths{s}} = \traces{\paths{s'}}$.
\end{definition}

\begin{remark} In the presence of non-totality of the transition relation
of a Kripke Structure, it no longer suffices to consider only the
infinite paths as the basis for defining trace equivalence. Instead,
\emph{maximal} paths are considered, which in addition to the infinite
paths, also contains paths made up of sequences of states that end in
a sink-state, \ie, a state without outgoing edges.
\end{remark}
For models in \LTS, we define similar-spirited concepts; for the
origins of the definition, we refer to Van Glabbeek's lattice of
equivalences~\cite{vanGlabbeek01}.

\newcommand{\runs}[1]{\ensuremath{\mathsf{Runs}(#1)}}
\newcommand{\bareruns}[1]{\ensuremath{\mathsf{Runs_b}(#1)}}
\begin{definition} Let $T = \tuple{S, \act,\ltstrans{}}$ be a Labelled
Transition System.  A \emph{run} starting in a state $s \in S$ is an
infinite, alternating sequence of states and actions $s_0\ a_0\ s_1\
a_1\ \ldots$ satisfying $s_i \ltstrans{a_i} s_{i+1}$ for all $i$, and
$s = s_0$.  The set of all runs starting in $s_0$ is denoted $\runs{s_0}$.

\end{definition}

\begin{definition} Let $T = \tuple{S, \act, \ltstrans{}}$ be a Labelled
Transition System.  The \emph{trace} of a run $\rho = s_0\ a_0\
s_1\ a_1\ \ldots$, denoted $\trace{\rho}$, is the infinite sequence
$a_0\ a_1\ \cdots$.  For a set of runs $R$, we define
$$\traces{R} = \{ \trace{\rho} \mid \rho \in R \}$$
States $s,s' \in S$ are \emph{completed trace equivalent}, denoted by
$T \models s \traceeq s'$ iff $\traces{\runs{s}} = \traces{\runs{s'}}$.

\end{definition}

\begin{theorem}
\label{th:ks_traceeq}
Let $K = \tuple{S, \AP, \kstrans, L}$ be a Kripke Structure. For all
$s,s' \in S$, we have $K \models s \traceeq s'$ iff
$\lts(K) \models s \traceeq s'$.
\end{theorem}

\begin{proof}
See Appendix~\ref{pf:th:lts_bisimilarity} for details.
\qed
\end{proof}
In a similar vein, we obtain that completed trace equivalence in
\LTS is preserved and reflected by trace equivalence in \KS.
\begin{theorem}
\label{th:lts_traceeq}
Let $T = \tuple{S, \act, \ltstrans{}}$ be a Labelled Transition System.
Let $s, s' \in S$ be arbitrary states. We have $T \models s \traceeq s'$ iff
$\ks(T) \models s \traceeq s'$.

\end{theorem}
\begin{proof}
Along the lines of the proof for Theorem~\ref{th:ks_traceeq}.\qed
\end{proof}

\section{Minimisations in \LTS and \KS}
\label{sec:minimisations}

As we concluded in Section~\ref{Sect:preliminaries}, the mappings $\lts$
and $\ks$ cannot be used to freely move to and fro the computational
models.  Instead, we introduce two additional mappings, \viz, $\ltsrev$
and $\ksrev$ that act as inverses to $\lts$ and $\ks$, respectively, and
we show that these can be used to come to our results for minimisation.
Here, we focus on the computationally most attractive
equivalences, \viz, \emph{bisimilarity} and \emph{stuttering equivalence}.

\renewcommand{\min}[1]{{#1}\textrm{-min}}
\newcommand{\minKS}[1]{{#1}\textrm{-min}_{\KS}}
\newcommand{\minLTS}[1]{{#1}\textrm{-min}_{\LTS}}
\newcommand{\quotient}[2]{\ensuremath{{#1}_{/#2}}}
Let $\sim {} \in \{\bisim, \stuttering\}$ and $\leftrightarrow {} \in \{ \bisim, \dsbbisim \}$ be arbitrary equivalence
relations on \KS and \LTS, respectively.  For a given model $K$ in \KS, its \emph{quotient}
with respect to $\sim$ is denoted $\quotient{K}{\sim}$. Similarly, for a given model $T$ in \LTS, its \emph{quotient}
with respect to $\leftrightarrow$ is denoted $\quotient{T}{\leftrightarrow}$. We assume
unique functions $\minKS{\sim}$ for \KS, and $\minLTS{\leftrightarrow}$
for \LTS that uniquely determine transition systems that are isomorphic
to the quotient. If, from the equivalence relation $\sim$, the setting
is clear, we drop the subscripts and write $\min{\sim}$ instead.

\subsection{Minimisation in \KS via minimisation in \LTS}

We first characterise a subset of models of \LTS for which we can
define our inverse $\ltsrev$ of $\lts$.

\begin{definition}
Let $T = \tuple{S, \act, \ltstrans{}}$ be a Labelled Transition System.
Then $T$ is reversible iff
\begin{enumerate}
\item $\act = 2^{\AP} \cup \{\bot\}$, for some set $\AP$;

\item for all $s,s'\in S$ and $a \in \act \cup \{\tau\}$,
if $s \ltstrans{a} s'$, then $s'\ltstrans{\bot}$;

\item for all $s,s',s'' \in S$ such that $s \ltstrans{\bot} s'$ and
$s \ltstrans{\bot} s''$, we require that
$s' \ltstrans{a}$ and $s'' \ltstrans{a'}$ implies $a = a'$ for all
actions $a,a' \in \act$.

\end{enumerate}
\end{definition}
Note that any embedding $\lts(K)$ of a Kripke Structure $K$
is a reversible Labelled Transition System. Reversibility
is preserved by the quotients for $\bisim$ and $\dsbbisim$,
as stated by the following proposition.
\begin{proposition}
\label{prop:reversibility}
Let $T$ be an arbitrary reversible Labelled Transition
System. Then $\quotient{T}{\leftrightarrow}$, for $\leftrightarrow {} \in \{\bisim, \dsbbisim\}$,
is reversible. \qed

\end{proposition}
The embedding $\lts$
introduces a fresh, \emph{a priori} known action label $\bot$.
We treat this constant differently from all other actions in
our reverse embedding.
\begin{definition}
\label{translationltsreverse}
Let $T = \tuple{S, \act, \ltstrans{}}$ be a reversible Labelled Transition
System.
We define the Kripke Structure $\ltsrev(T)$ as the structure
$\tuple{S', \AP, \kstrans, L}$, where:
\begin{itemize}
\item $S' = \{s \in S ~|~ s \ltstrans{\bot} \}$;

\item $\AP$ is such that $\act = 2^{\AP} \cup \{\bot\}$;

\item $\kstrans$ is the least relation satisfying the single rule:
$$
\sosrule{s \ltstrans{a} s' \qquad a \not=\bot \qquad s \ltstrans{\bot}}
        {s \kstrans s'}
$$

\item $L(s) = a$ for the unique $a$ such that $s \ltstrans{\bot} s'
\ltstrans{a}$.

\end{itemize}

\end{definition}

\noindent
The following proposition establishes that $\ltsrev$ is the
inverse of embedding $\lts$.

\newcommand{\id}{\ensuremath{\mathsf{Id}}}
\begin{proposition}
\label{prop:lts_inverse}
We have $\ltsrev \circ \lts = \id$.\qed
\end{proposition}

\begin{proof}
Establishing the isomorphism follows immediately from the definitions
and the observation that $\lts(K)$ is reversible. See Appendix \ref{pf:prop:lts_inverse}.\qed
\end{proof}

\newcommand{\simKS}{\sim_{\ks}}
\newcommand{\simLTS}{\sim_{\lts}}

Note that reversibility of a Labelled Transition System $T$ is too
weak to obtain $(\lts \circ \ltsrev)(T) = T$, as the following example
illustrates:
\begin{example}
Consider the Labelled Transition System left below.
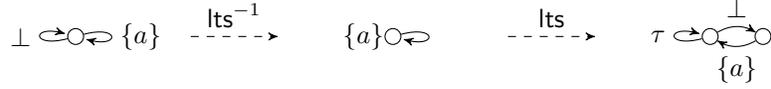
\begin{center}
\begin{tikzpicture}[>=stealth',node distance=40pt]
\tikzstyle{every state}=[inner sep=1pt, minimum size=6pt];

\node[state] (t) {};

\node (t') [right of=t] {};
\node (t'') [right of=t'] {};

\node[state] (u) [right of=t''] {};

\node (u') [right of=u] {};
\node (u'') [right of=u'] {};

\node[state] (v) [right of=u''] {};
\node[state] (v') [right of=v,xshift=-20pt] {};

\draw[->] (t) edge [loop left] node [left] {$\bot$} (t)
          (t) edge [loop right] node [right] {$\{a\}$} (t)
          (u) edge [loop right] (u)
          (v) edge [loop left] node [left] {$\tau$} (v)
          (v) edge [bend left] node [above] {$\bot$} (v')
          (v') edge [bend left] node [below] {$\{a\}$} (v)
;

\draw (u) node [left] {$\{a\}$};
\draw[->,dashed] (t') edge node [above] {$\ltsrev$} (t'')
                 (u') edge node [above] {$\lts$} (u'');

\end{tikzpicture}
\end{center}
Clearly, the
Labelled Transition System
is reversible, so the mapping $\ltsrev$ is applicable. Its result is
given by the Kripke Structure in the middle. Applying $\lts$ to the middle Kripke Structure
yields the Labelled Transition System at the right. It is clear that the latter is not isomorphic
to the original Labelled Transition System.\qed
\end{example}

\begin{lemma}
\label{lem:lts_bisim} We have $\minLTS{\bisim} \circ \lts \circ \minKS{\bisim}
= \lts \circ \minKS{\bisim}$.

\end{lemma}

\begin{proof}
See Appendix \ref{pf:lem:lts_bisim}.\qed
\end{proof}

\begin{lemma}
\label{lem:lts_stut}
We have $\minLTS{\dsbbisim} \circ \lts \circ \minKS{\stuttering} =
\lts \circ \minKS{\stuttering}$.
\end{lemma}

\begin{proof}
See Appendix \ref{pf:lem:lts_stut}.\qed
\end{proof}
Before we present the main theorems concerning the minimisations in \KS
through minimisations in \LTS, we first show that it suffices to prove
such results for Kripke Structures that are already minimal; see the
lemma below.
\begin{lemma}
\label{lem:minimal}
Let $\sim {} \in \{\bisim,\stuttering\}$ and
$\leftrightarrow {} \in \{\bisim,\dsbbisim \}$ such that $\lts$ preserves and reflects $\sim$ through
$\leftrightarrow$.  Then
$$
\begin{array}{ll}
& \min{\sim} = \ltsrev \circ \min{\leftrightarrow} \circ \lts \circ
\min{\sim} \\
\text{implies } & \\
& \min{\sim} = \ltsrev \circ \min{\leftrightarrow} \circ \lts
\end{array}
$$
\end{lemma}

\begin{proof}
Assume that we have
\begin{equation}
\tag{*}
\label{eq:*}
\min{\sim} = \ltsrev \circ \min{\leftrightarrow} \circ \lts \circ
\min{\sim}
\end{equation}
By definition of $\min{\sim}$, we find
$ \forall K: \min{\sim}(K)\ \sim\ K$.
Since, by assumption, $\lts$ preserves and reflects $\sim$ through
$\leftrightarrow$, we derive
$
\forall K: \lts(K) \leftrightarrow \lts(\min{\sim}(K))
$.
By definition of $\min{\leftrightarrow}$, this means that we have:
$$
\min{\leftrightarrow} \circ \lts\ = \
           \min{\leftrightarrow} \circ \lts \circ \min{\sim}
$$
As $\ltsrev$ is functional, and $\min{\leftrightarrow}$ preserves
reversibility, we immediately obtain:
\begin{equation}
\tag{**}
\label{eq:**}
\ltsrev \circ \min{\leftrightarrow} \circ \lts =
\ltsrev \circ \min{\leftrightarrow} \circ \lts \circ \min{\sim}
\end{equation}
The desired conclusion then follows by combining~\ref{eq:*} and~\ref{eq:**}.
\qed
\end{proof}
We finally state the two main theorems in this section.
\begin{theorem}
\label{th:ks2lts_bisim_minimal}
We have $\minKS{\bisim} = \ltsrev \circ \minLTS{\bisim} \circ \lts$.
\end{theorem}

\begin{proof}
Lemma~\ref{lem:lts_bisim} guarantees
\[
\minLTS{\bisim} \circ \lts \circ \minKS{\bisim}
=
\lts \circ \minKS{\bisim}
\]
Functionality of $\ltsrev$, combined with Proposition~\ref{prop:reversibility},
we find:
\[
\ltsrev \circ \minLTS{\bisim} \circ \lts \circ \minKS{\bisim}
=
\ltsrev \circ \lts\circ \minKS{\bisim}
\]
By Lemma~\ref{lem:minimal}, we then have our desired conclusion:
$$
\minKS{\bisim}
=
\ltsrev \circ \minLTS{\bisim} \circ \lts
$$
\qed
\end{proof}

\begin{theorem}
\label{th:ks2lts_stut_minimal}
We have  $\minKS{\stuttering}
= \ltsrev \circ \minLTS{\dsbbisim} \circ \lts$.
\end{theorem}

\begin{proof}
Similar to
Theorem~\ref{th:ks2lts_bisim_minimal}, using Lemma~\ref{lem:lts_stut}
instead of Lemma~\ref{lem:lts_bisim}. \qed
\end{proof}

\subsection{Minimisation in \LTS via minimisation in \KS}

In the previous section, we showed that one can minimise in \KS
with respect to bisimilarity or stuttering equivalence,
using the embedding $\lts$, a matching equivalence relation in \LTS and
converting to \KS again. In a similar vein, we propose a reverse
translation for $\ks$, which allows one to return to \LTS from
\KS. We first characterise a set of Kripke Structures that are amenable
to translating to Labelled Transition Systems.

\begin{definition} Let $K = \tuple{S, \AP, \kstrans, L}$ be a Kripke
Structure. Then $K$ is reversible iff
\begin{enumerate}
\item $\AP = \act \cup \{\bot\}$ for some set $\act$;
\item $|L(s)| = 1$ for all $s \in S$;
\item for all $s$ for which $\bot \notin L(s)$, we require that
for all $s', s''$, $s \kstrans s'$ and $s \kstrans s''$ implies both $s' = s''$ and
$L(s') = \{\bot\}$.

\end{enumerate}

\end{definition}
\begin{proposition}
\label{prop:reversibility2}
Let $K$ be an arbitrary reversible Kripke Structure.
Then $\quotient{K}{\sim}$, for $\sim \in \{\bisim, \stuttering\}$,
is reversible.

\end{proposition}

\begin{definition}
Let $K = \tuple{S, \AP, \kstrans, L}$ be a reversible Kripke Structure.
The Labelled Transition System $\ksrev(K)$ is the structure
$\tuple{S', \act, \ltstrans{}}$, where:
\begin{itemize}
\item $S' = \{s \in S ~|~ L(s) = \{\bot\} \}$;

\item $\act$ is such that $\act = \AP \setminus \{\bot\}$;

\item $\ltstrans{}$ is the least relation satisfying:
$$
\begin{array}{cp{1cm}c}
\sosrule{s \kstrans s' \qquad L(s) = L(s')}
        {s \ltstrans{\tau} s'}
& &
\sosrule{s \kstrans s'' \qquad a \in L(s'')\setminus \{\bot\} \qquad s'' \kstrans s'}
        {s \ltstrans{a} s'}
\end{array}
$$

\end{itemize}
\end{definition}
\begin{proposition}
\label{prop:ks_inverse}
We have $\ksrev \circ \ks = \id$.
\end{proposition}

\begin{proof}
Similar to the proof of Proposition \ref{prop:lts_inverse}.\qed
\end{proof}

\noindent
Without further elaboration, we state the final results.
\begin{theorem}
\label{th:lts2ks_bisim_minimal}
We have $\minLTS{\bisim} = \ksrev \circ \minKS{\bisim} \circ \ks$.\qed
\end{theorem}

\begin{theorem}
\label{th:lts2ks_stut_minimal}
We have $\minLTS{\dsbbisim} = \ksrev \circ \minKS{\stuttering} \circ \ks$.
\qed
\end{theorem}

\section{Conclusions}\label{Sect::conc}

Our results in Section~\ref{Sect:preservations} naturally extend
the fundamental results obtained by De Nicola and Vaandrager
in~\cite{DBLP:conf/litp/NicolaV90,DBLP:journals/jacm/NicolaV95}.
In a sense, we can now state that their embeddings
$\ks$ and $\lts$ are canonical for four commonly used equivalence
relations.

While the stated embeddings have traditionally been used to come to
results about the correspondence between logics, the question whether
they support minimisation modulo behavioural equivalences was never
answered. Thus, in addition to the above stated results, we proved that
indeed the embeddings $\ks$ and $\lts$ can serve as basic tools in the
problem of minimising modulo a behavioural equivalence relation. To this
end, we defined inverses of the embeddings to compensate for the fact that
composing $\ks$ and $\lts$ does not lead to transition systems that are
comparable (in whatever sense) to the one before applying the embeddings.
The latter results are clearly interesting from a practical perspective,
allowing one to take full advantage of state-of-the-art minimisation
tools available for one computational model, when minimising in the other.

Our minimisation results are for two of the most commonly used equivalence
relations that are, arguably, still efficiently computable. However,
we do intend to extend our results also in the direction of (completed)
trace equivalence and similarity. As a slightly more esoteric research
topic, one could look for improving on the embedding $\lts$, as, compared
to the embedding $\ks$, it introduces more ``noise''. For instance, it
yields Labelled Transition Systems that have runs that cannot sensibly
be related to paths in the original Kripke Structure.

\bibliographystyle{plain}
\bibliography{lit}

\cleardoublepage
\appendix
\section{Proofs for Section \ref{Sect:preservations}}

\subsection{Proof of Theorem \ref{th:lts_similarity}}
\label{pf:th:lts_similarity}

Consider states $s$ and $s'$ in a Kripke structure $\tuple{S,A,\kstrans,L\rangle}$. Assume that $K \models s \simm s'$ and that this is witnessed by the simulation relations $B$ and $C$ with $(s,s') \in B$ and $(s',s) \in C$.
We show that, with respect to the Labelled Transition System associated with the Kripke structure, the relation $B' = B \cup \{ (\shadow{s},\shadow{s}') \mid (s,s') \in B \}$ is a simulation relation with  $(s, s') \in B$. In a similar way a simulation relation $C'$ with $(s',s) \in C'$ can be defined. This part is omitted.

First consider an arbitrary pair $(\shadow{s},\shadow{s}') \in B'$. This is due to the fact that $(s,s') \in B$. By construction the only transitions for $\shadow{s}$ and $\shadow{s}'$ are $\shadow{s} \ltstrans{L(s)} s$ and $\shadow{s}' \ltstrans{L(s')} s'$. From the fact that $(s,s') \in B$ it follows that $L(s) = L(s')$. This suffices to satisfy all transfer conditions for the pair $(\shadow{s},\shadow{s}')$.

Next, consider an arbitrary pair $(s,s') \in B'$. This is due to the fact that $(s,s') \in B$. Let us consider all transitions from $s$.
\begin{itemize}
\item $s \ltstrans{\bot} \shadow{s}$. Since $s' \in S$ we have $s' \ltstrans{\bot} \shadow{s}'$. Since $(s,s') \in B$ it also follows that $(\shadow{s},\shadow{s}') \in B'$.

\item $s\ltstrans{\tau} t$ for some $t \in S$ such that $s \kstrans t$ and $L(s) = L(t)$.
Since $(s,s') \in B$ and $B$ is a simulation, it follows that $L(s)=L(s')$ and $s' \kstrans t'$ for some $t' \in S$ such that $(t,t') \in B$. Since $(t,t') \in B$ we have $L(t) = L(t')$, and therefore $L(s') = L(t')$ as well. Thus, by construction $s' \ltstrans{\tau} t'$. From $(t,t') \in B$ we obtain $(t,t') \in B'$.

\item $s\ltstrans{L(t)} t$ for some $t\in S$ such that $s \kstrans t$ and $L(s) \neq L(t)$.
Since $(s,s') \in B$ and $B$ is a simulation, it follows that $L(s)=L(s')$ and $s' \kstrans t'$ for some $t' \in S$ such that $(t,t') \in B$. Since $(t,t') \in B$ we have $L(t) = L(t')$, and therefore $L(s') \neq L(t')$ as well. Thus, by construction $s' \ltstrans{L(t')} t'$. From $(t,t') \in B$ we obtain $(t,t') \in B'$. \qed
\end{itemize}

\subsection{Proof of Theorem \ref{th:ks_similarity}}
\label{pf:th:ks_similarity}

Consider states $s$ and $s'$ in a Labelled Transition System $T=\tuple{
S,A,\ltstrans{}}$. Assume that $T \models s \simm s'$ and that this
is witnessed by the simulations $B$ and $C$ with $(s,s') \in B$
and $(s',s)\in C$. We show that, with respect to the Kripke structure
associated with the Labelled Transition System, the relation $B' = B \cup
\{ ((s,a,t),(s',a,t')) \mid (s,a,t),(s',a,t') \in S' \land (s,s'),(t,t')
\in B \}$ is a simulation relation with $(s,s') \in B$. Here $S'$ is
the set of states of that Kripke structure as prescribed by Definition
\ref{translationks}. Similarly, a simulation relation $C'$ with $(s',s)
\in C$ can be defined.

First consider a pair $(s,s')$ that is present in $B'$ due to its presence in $B$. Since $s,s'\in S$ we have by definition that $L(s) = \{\bot\} = L(s')$.
We consider all possible transitions from $s$. By construction the only possible transitions for $s$ are the following.

\begin{itemize}
\item $s \kstrans t$ for some $t \in S$ such that $s \ltstrans{\tau} t$. Since $(s,s') \in B$ and $B$ is a simulation relation, we have $s' \ltstrans{\tau} t'$ for some $t' \in S$ such that $(s',t') \in B$. By construction then also $s' \kstrans t'$ and $(s',t') \in B'$.
\item $s \kstrans (s,a,t)$ for some $a \in A$ and $t \in S$ such that $a \neq \tau$ and $s \ltstrans{a} t$. Since $(s,s') \in B$ and $B$ is a simulation relation, we have $s' \ltstrans{a} t'$ for some $t' \in S$ such that $(t,t') \in B$. By construction then also $s' \kstrans (s',a,t')$. Note that $((s,a,t),(s',a,t')) \in B'$ since $(s,s') \in B$ and $(t,t')\in B$.
\end{itemize}

Next, consider a pair $((s,a,t),(s',a,t'))$ that is present in $B'$ due to presence of both $(s,s')$ and $(t,t')$ in $B$. By construction $L((s,a,t)) = \{ a \} = L((s',a,t'))$. Let us consider all transitions from $(s,a,t)$.
The only possible transition is $(s,a,t) \kstrans t$. Since $(s,s') \in B$, $s \ltstrans{a} t$ and $B$ is a simulation relation it follows that $s' \ltstrans{a} t'$ for some $t' \in S$ such that $(t,t') \in B$. By construction then also $(s',a,t') \kstrans t'$ and $(t,t') \in B'$. \qed

\subsection{Proof of Theorem \ref{th:lts_bisimilarity}}
\label{pf:th:lts_bisimilarity}

Consider states $s$ and $s'$ in a Kripke structure $\tuple{S,A,\kstrans,L}$. Assume that $K \models s \bisim s'$ and that this is witnessed by the bisimulation relation $B$ with $(s,s') \in B$. Thus $B$ is a simulation relation with $(s,s') \in B$ and with $(s',s) \in B$.
We define the relation $B' = B \cup \{ (\shadow{s},\shadow{s}') \mid (s,s') \in B \}$. It follows from the proof of Theorem \ref{th:lts_similarity} that $B'$ is a simulation relation for $(s,s')$ and for $(s',s)$. \qed

\subsection{Proof of Theorem \ref{th:ks_traceeq}}
\label{pf:th:ks_traceeq}

Before we prove Theorem \ref{th:ks_traceeq} in this section, we establish
an intermediate result concerning the relation between paths ---and
their prefixes--- of a Kripke Structure and the subset of bare runs,
defined below ---and their prefixes--- in the \LTS-embedding of the same
Kripke Structure.

\begin{definition}[Bare run] A run $\rho$ is said
to be a \emph{bare run} iff the labels occurring on the run differ from
$\bot$. The set of bare runs starting in a state $s$, for $s \ltstrans{\bot}$
is given by the set $\bareruns{s}$.
\end{definition}
\newcommand{\prefruns}[1]{\ensuremath{\mathsf{Runs}^p(#1)}}
\newcommand{\prefbareruns}[1]{\ensuremath{\mathsf{Runs}^p_b(#1)}}
\newcommand{\bare}[1]{\ensuremath{\beta(#1)}}

Let $T = \tuple{S, \act, \ltstrans{}}$.  Let $\prefruns{s} \subseteq S
(\act\ S)^*$ be the set of prefixes of runs starting in states $s \in S$;
likewise, $\prefbareruns{s} \subseteq S
((\act\setminus\{\bot\})\ S)^*$ is the set of prefixes of bare runs
starting in $s \in S$ satisfying $s \ltstrans{\bot}$.

Given a (finite) trace $\sigma \in \traces{\prefruns{s}}$, we write
$s \ltstrans{\sigma} t$ if there is some $\rho_p \in \prefruns{s}$
ending in state $t$ such that $\sigma = \trace{\rho_p}$.

\begin{definition} Let $\sigma \in \traces{\prefruns{s}}$.
Denote the sequence $\bare{\sigma}$, obtained from $\sigma$
by deleting
\begin{itemize}
\item all subsequences of the form $\bot\ l$;
\item $\bot$ in case $\sigma$ ends
as such.
\end{itemize}
\end{definition}
It is not hard to see that $\bare{\sigma} \in \traces{\prefruns{s}}$
implies $\bare{\sigma} \in \traces{\prefbareruns{s}}$, \ie,
any trace $\bare{\sigma}$ is generated by a bare run.
\begin{lemma}
\label{lem:bare_vs_ordinary}
For all $\sigma \in \traces{\prefruns{s}}$, $s \in S$ such that
$s \ltstrans{\bot}$ and all $t \in S \cup \shadow{S}$, we have
$s \xrightarrow{\sigma} t$ iff
\begin{enumerate}
\item
$s \xrightarrow{\bare{\sigma}} t$ and
$\forall \sigma': \sigma \not= \sigma'\bot$, \emph{or}

\item $s \xrightarrow{\bare{\sigma}\bot} t$ and
$\exists \sigma': \sigma = \sigma'\bot$.

\end{enumerate}

\end{lemma}
\begin{proof} By induction on the length of $\sigma$.
\end{proof}
The above lemma firmly establishes a connection between a trace $\sigma$
of a run starting in a state $s$ and the trace $\bare{\sigma}$. Intuitively,
as bare runs only pass through states that can perform a $\bot$
transition, any trace generated by a bare run can be ``pumped up'' to generate
an arbitrary trace that can lead to the same state as its corresponding
bare run, simply by following the loop $\bot\ l$, for some action label $l$.

\begin{lemma}
\label{lem:bareruns1}
Let $K = \tuple{S, \AP, \kstrans,  L}$ be a Kripke Structure.
Then
$s_0\ s_1\ \cdots \in \paths{s_0}$ implies
$s_0\ l_0\ s_1\ l_1\ \cdots \in \bareruns{s_0}$ in $\lts(K)$
for precisely one
infinite sequence $l_0\ l_1\ \cdots$. \emph{Vice versa,}
if for some infinite sequence $l_0\ l_1\ \cdots$, we have
$s_0\ l_0\ s_1\ l_1\ \cdots \in \bareruns{s_0}$ in $\lts(K)$,
then $s_0\ s_1\ \cdots \in \paths{s_0}$.
\end{lemma}
\begin{proof}
Follows by definition of $\lts$.
\end{proof}
Informally, the above lemma states that for each path in a Kripke
Structure, there is a unique matching bare run in its \LTS embedding,
and, \emph{vice versa}, for every bare run in its \LTS embedding, there
is a unique path in the Kripke Structure.

We next establish that the embedding $\lts$ is such that for the trace
equivalence of two states in a Labelled Transition System resulting from
the embedding $\lts$, it suffices to prove that the traces of all bare runs
coincide.  Formally, we have:

\begin{lemma}
\label{lem:bareruns2}
Let $K = \tuple{S, AP, \kstrans, L}$ be a Kripke Structure. Let $s,
s' \in S$, with $L(s) = L(s')$ be such that $\traces{\bareruns{s}} =
\traces{\bareruns{s'}}$ in $\lts(K)$. Then also $\traces{\runs{s}} =
\traces{\runs{s'}}$ in $\lts(K)$.

\end{lemma}

\begin{proof} By induction on the length of the traces.
\end{proof}
Since all runs are infinite, in the limit, any trace $\sigma \in
\traces{\runs{s}}$ is also in the set $\traces{\runs{s'}}$.

\begin{proof}[Theorem~\ref{th:ks_traceeq}]
Let states $s,s' \in S$ in a Kripke Structure be trace equivalent.
Suppose $\pi = s_0\ s_1\ s_2\ \ldots \in \paths{s}$ and $\pi' =
s_0'\ s_1'\ s_2'\ \ldots \in \paths{s'}$ are such that $\trace{\pi} =
\trace{\pi'}$. Because of Lemma~\ref{lem:bareruns1}, we find that there
must be unique $\rho \in \bareruns{s}$ and $\rho' \in \bareruns{s'}$
passing through the exact same states as the paths $\pi$ and $\pi'$
respectively. That is:
$$
\left \{
\begin{array}{ll}
\rho = s_0\ l_0\ s_1\ l_1\ s_2\ \ldots \\
\rho' = s_0'\ l_0'\ s_1'\ l_1'\ s_2'\ \ldots
\end{array}
\right .
$$
By construction of $\lts$, we have $l_i = \tau$ if
$L(s_i) = L(s_{i+1})$ and $l_i = L(s_{i+1})$ otherwise (and similarly for
$l_i'$). But from
the fact that $\trace{\pi} = \trace{\pi'}$, we find that $L(s_i)
= L(s'_i)$ for all $i$. Hence, also $l_i = l_i'$ for all $i$. But
this means that $\trace{\rho} = \trace{\rho'}$. Appealing to
Lemma~\ref{lem:bareruns1}, \emph{all} bare runs correspond to paths
in the Kripke Structure. Hence, we find that
$\traces{\bareruns{s}} = \traces{\bareruns{s'}}$. Since $L(s) =
L(s')$, Lemma~\ref{lem:bareruns2} yields the desired conclusion that
$\traces{\runs{s}} = \traces{\runs{s'}}$.

For the other direction, we assume that states $s,s' \in S$ are
trace equivalent in $\lts(K)$. In short, this means that the set of
bare runs starting in $s$ and $s'$ produce the same traces. Let
$\rho = s_0\ l_0\ s_1\ l_1\ \ldots$ be a bare run starting in $s$,
and $\rho' = s_0'\ l_0'\ s_1'\ l_1'\ \ldots$ be a bare run starting in
$s'$, such that $l_i = l_i'$. Using Lemma~\ref{lem:bareruns1}, we find
that there are unique matching paths $\pi = s_0\ s_1\ s_2\ \ldots$
and $\pi' = s_0'\ s_1'\ s_2'\ \ldots$.
By construction of $\lts$, we find that
this implies that $L(s_i) = L(s_i')$ for all $i$ satisfying that $i \ge j$
for the least $j$ such that $l_j \not= \tau$. For all $0 <i < j$, we
observe that $l_i = l_{i-1} = \tau$, which can only be if
$L(s_i) = L(s_{i-1})$ for $0 < i < j$. Likewise, $L(s_i') = L(s'_{i-1})$.
Since \emph{all} traces starting in $s$ and $s'$ are the same in
$\lts(K)$, also $\trace{s\ \bot\ L(s)\ \rho} = \trace{s'\ \bot\ L(s')\ \rho'}$,
which can only be the case when $L(s) = L(s')$. But then also
$L(s_i) = L(s'_i)$ for all $i \ge 0$. Hence, $\trace{\pi} = \trace{\pi'}$.
Since all paths starting in
$s$ and $s'$ correspond to unique bare runs starting in $s$ and $s'$
in $\lts(K)$, this means we have considered all possible paths and
therefore all possible traces.
\end{proof}

\section{Proofs for Section \ref{sec:minimisations}}

\subsection{Proof of Proposition \ref{prop:lts_inverse}}
\label{pf:prop:lts_inverse}

This theorem follows directly from the definitions.
Consider arbitrary Kripke structure $K = \tuple{S,AP,\kstrans,L}$. Let $\lts(K) = T= \tuple{ S',A',\ltstrans{}}$ and $\ltsrev(T) = K'=\tuple{S'',AP',\kstrans',L'}$. We will show that $S''=S$, $AP'=AP$, $\kstrans'=\kstrans$ and $L'=L$, thus establishing the isomorphism of $K$ and $K'$.

From the definition of $\lts$ (applied to $K$) it follows that
\begin{itemize}
\item $S'= S \cup \{ \shadow{s} \mid s \in S \}$;
\item $A' = 2^{AP} \cup \{ \bot \}$;
\item $\ltstrans = \begin{array}[t]{l}
\{ (s,\bot,\shadow{s}), (\shadow{s},L(s),s) \mid s,s'\in S \} \cup ~ \\
\{ (s,\tau,s') \mid s,s'\in S \land L(s)=L(s') \land s \kstrans s' \} \cup ~ \\
\{ (s,L(s'),s') \mid s,s'\in S \land L(s) \neq L(s') \land s \kstrans s' \}
\end{array}$
\end{itemize}
and application of $\ltsrev$ (applied to $T$) gives
\begin{itemize}
\item $S''= \{ s' \mid s' \in S' \land s' \ltstrans{\bot} \} = \{ s' \mid s' \in S \cup \{ \shadow{s} \mid s \in S \} \land s' \ltstrans{\bot} \}$. Since $s'\ltstrans{\bot}$ iff $s'\in S$ we obtain
    $S'' = \{ s' \mid s' \in S \} = S$.
\item $AP' = AP$.
\item $\kstrans'
\begin{array}[t]{l}
= \{ (s',t') \mid (s',a,t') \in \ltstrans{} \land a \neq \bot \land s' \ltstrans{\bot} \} \\
= \{ (s',t') \mid (s',a,t') \in \ltstrans{} \land a \neq \bot \land s' \in S \} \\
= \{ (s',t') \mid s', t' \in S \land L(s') =L(t') \land s' \kstrans t' \} \\
\cup \ \{ (s',t') \mid s', t' \in S \land L(s') \neq L(t') \land s' \kstrans t' \} \\
= \{ (s',t') \mid s', t' \in S \land s' \kstrans t' \} \\
= \kstrans
\end{array}$
\item $L'(s') = a$ where $a$ is such that $s' \ltstrans{\bot} t' \ltstrans{a}$ for some $t'$. Therefore $s'\in S$ and $t'= \shadow{s'}$. From this it follows that $a = L(s')$. So $L'(s')=L(s')$. \qed
\end{itemize}

\subsection{Proof of Lemma \ref{lem:lts_bisim}}
\label{pf:lem:lts_bisim}

Consider a Kripke structure $K = \tuple{S,AP,\kstrans,L}$ that is minimal w.r.t.\ strong bisimilarity (on \KS). We have to show that $\lts(K)$ is minimal w.r.t.\ strong bisimilarity (on \LTS). We show that (1) the identity relation on the states of $\lts(K)$ is a bisimulation relation, and (2) that this bisimulation relation is maximal.

We know, since $K$ is minimal, that the identity relation on $S$ is a maximal bisimulation relation. From this it follows that the identity relation on $S'$ (the states of $\lts(K)$) is a bisimulation relation as well.

Now assume that the identity relation on $S'$ is not the maximal bisimulation relation, i.e.,  there exists a bisimulation relation $B \subseteq S'\times S'$ that relates at least one pair of different states. First, we show that it has to be the case that at least one pair of different states from $S$ is related by $B$.

This can be seen as follows. Consider a pair of different states $s$ and $t$ related by $B$. Suppose that $s \in S$ and $t \not\in S$. In this case, by definition of $\lts$, $s \ltstrans{\bot}$, but $t \ltsntrans{\bot}$. Hence $s$ and $t$ cannot be related by a bisimulation relation. The case that $s\not\in S$ and $t \in S$ is similar.
In case both $s \not\in S$ and $t \not\in S$, by definition $s = \shadow{s'}$ and $t = \shadow{t'}$ for some $s',t'\in S$ with $s'\neq t'$. Then, by definition of $\lts$, the only transitions of $s$ and $t$ are $s \ltstrans{L(s')} s'$ and $t \ltstrans{L(t')} t'$. In order for $s$ and $t$ to be related by $B$ necessarily $s'$ and $t'$ need to be related by $B$. Thus we can safely conclude that $B$ relates a pair of different states $s$ and $t$, both from $S$.

Now we show that $B \cap (S \times S)$ is a bisimulation relation on \KS, thus contradicting the assumption that the identity relation on $S$ is the maximal bisimulation relation.

Consider a pair of different states $s$ and $t$, both from $S$, that are related by $B$. We show that $L(s) = L(t)$.
 This follows from the following observations. Both $s$ and $t$ each have a single $\bot$-transition: $s \ltstrans{\bot} \shadow{s}$ and $t \ltstrans{\bot} \shadow{t}$. Then, also $\shadow{s}$ and $\shadow{t}$ are related by $B$. These states each have only one transition: $\shadow{s} \ltstrans{L(s)} s$ and $\shadow{t} \ltstrans{L(t)} t$. From this it follows that $L(s) =L(t)$.

 Assume that $s \kstrans s'$ for some $s'\in S$. We distinguish two cases:
 \begin{itemize}
 \item $L(s) = L(s')$. Then $s \ltstrans{\tau} s'$. Then $t \ltstrans{\tau} t'$ for some $t'$ such that $(s',t') \in B$. Since $B$ cannot relate states from $S$ (such as $s'$) with states outside $S$, also $t'\in S$. Therefore, by definition of $\lts$, $t \kstrans t'$.
 \item  $L(s) \neq L(s')$. Then $s \ltstrans{L(s')} s'$. Then $t \ltstrans{L(s')} t'$ for some $t'$ such that $(s',t') \in B$. Since $B$ cannot relate states from $S$ (such as $s'$) with states outside $S$, also $t'\in S$. Then, by definition of $\lts$ it has to be the case that $L(s') = L(t')$ and $t \kstrans t'$.
 \end{itemize}
 In each case it follows that $t \kstrans t'$ and $s'$ and $t'$ are related by $B \cap (S \times S)$., which was to be shown.

The case that $t \kstrans t'$ for some $t'\in S$ needs to be mimicked is similar.

From the contradiction obtained it can be concluded that the identity relation on $S'$ is the maximal bisimulation relation.

\subsection{Proof of Lemma \ref{lem:lts_stut}}
\label{pf:lem:lts_stut}

Consider a Kripke structure $K = \tuple{S,AP,\kstrans,L}$ that is minimal w.r.t.\ stuttering equivalence (on \KS). We have to show that $\lts(K)$ is minimal w.r.t.\ divergence-sensitive branching bisimilarity (on \LTS). We show that (1) the identity relation on the states of $\lts(K)$ is a divergence-sensitive branching bisimulation relation, and (2) that this bisimulation relation is maximal.

We know, since $K$ is minimal, that $K_d$ is minimal with respect to
divergence-blind stuttering equivalence. Denote the states of $K_d$ by
$S \cup \{ s_d\}$. Hence, the identity relation on $S\cup\{s_d\}$ is a
maximal divergence-blind stuttering bisimulation relation with respect
to the Kripke structure $K_d$. From this it follows that the identity
relation on $S'$ (the states of $\lts(K)$) is a divergence-sensitive
branching bisimulation relation as well.

Now assume that the identity relation on $S'$ is not the maximal bisimulation relation, i.e.,  there exists a divergence-sensitive branching bisimulation relation $B'$ such that there are different states $s$ and $t$ from $S'$ with $(s,t) \in B'$. We distinguish four cases:
\begin{itemize}
\item $s \in S$ and $t \not\in S$. In this case, by definition of $\lts$, $s \ltstrans{\bot}$, but $t \ltsntrans{\bot}$ and $t \ltsntrans{\tau}$. Therefor the transition from $s$ cannot be mimicked from $t$. So this case cannot occur.

\item $s \not\in S$ and $t \in S$. Similar to the previous case.

\item $s \in S$ and $t \in S$. We have to show that there exists a divergence-blind stuttering bisimulation relation $B''$ with $(s,t) \in B''$.

First we consider the case that $s \kstrans s'$ for some $s'\in S \cup \{ s_d \}$. We can distinguish two cases
\begin{itemize}
\item Suppose that $s'\in S$. By definition $s \ltstrans{a} s'$. Then, by definition of divergence-sensitive branching bisimulation, we have $a=\tau$ and $(s',t) \in B'$, or the existence of $t_i$ and $t'$ such that
    \[ t \ltstrans{\tau} \cdots \ltstrans{\tau} t_i \ltstrans{\tau} \cdots \ltstrans{\tau} t_n \ltstrans{a} t'\]
     with $(s,t_i) \in B'$ (using the Stuttering Lemma) and $(s',t') \in B'$. In the first case we have $(s',t) \in B'$ and in the second case we have \[ t \kstrans \cdots \kstrans t_i \kstrans \cdots \kstrans t_n \kstrans t'\]  with $(s,t_i)\in B'$ and $(s',t') \in B'$.
\item Suppose that $s'=s_d$. By definition there is an infinite sequence
    \[ s \kstrans \cdots \kstrans s_i \kstrans \cdots \]
    of states with the same label. Therefore, in $\lts(K)$ there is an infinite sequence
    \[ s \ltstrans{\tau} \cdots \ltstrans{\tau} s_i \ltstrans{\tau} \cdots \]
    where all states have the same label $L(s) = L(s_i)$.
    Hence, there is an infinite sequence
    \[ t \ltstrans{\tau} \cdots \ltstrans{\tau} t_j \ltstrans{\tau} \cdots \ . \]
    Therefore, in $K$, there is an infinite sequence
    \[ t \kstrans \cdots \kstrans t_j \kstrans \cdots \]
    where $L(t) = L(t_j)$ for all $j$.
    Thus $t \kstrans s_d$ as required.
\end{itemize}

\item $s \not\in S$ and $t \not\in S$. By definition the only transition of $s$ is of the form $s \ltstrans{L(s')} s'$ for some $s'\in S$. Since $t \ltsntrans{\tau}$ obviously the only way to mimic the transition is by means of $t \ltstrans{L(s')} t'$ for some $t' \in S$ with $L(t') = L(s')$. Necessarily $(s',t') \in B$. We have established in the previous item that such $s'$ and $t'$ cannot be related. Therefore, also $s$ and $t$ cannot be related.
\end{itemize}

Second, we show that $L(s) = L(t)$. Since $s \in S$ we have $s \ltstrans{\bot} \shadow{s} \ltstrans{L(s)} s$. Then, $t \ltstrans{\tau^*} t^* \ltstrans{\bot} t'$ and $t'\ltstrans{\tau^*} t^{**} \ltstrans{L(s)} t''$ with $(s,t^*) \in B'$, $(\shadow{s},t') \in B'$, $(\shadow{s},t^{**})\in B'$ and $(s,t'') \in B'$.
It follows that $L(t) = L(t^*)$ and from the fact that $t'=\shadow{t}$ it follows that $L(t') = L(t)$ as well. Similarly, $L(t^{**}) = L(t')$. Since $t''=\shadow{t^{**}}$ it also follows that $L(s) = L(t'') = L(t^{**})$. Thus we have obtained $L(s)=L(t)$.

We have shown that $K_d$ was not minimal. Therefore the assumption that $\lts(K)$ is not minimal is flawed, which completes the proof.

\end{document}